\documentclass[11pt]{article}
\usepackage[in]{fullpage}

\usepackage[shortlabels]{enumitem}
\usepackage{bm}
\usepackage{amsfonts}
\usepackage{amsmath}
\usepackage{amssymb}
\usepackage{xcolor}
\usepackage{mathtools}
\usepackage{graphicx}
\usepackage{amsthm}
\usepackage{qtree}
\usepackage{tree-dvips}
\usepackage{float}
\usepackage{hyperref}
\usepackage{footnote}
\usepackage[nameinlink]{cleveref}
\usepackage{braket}
\usepackage{mathrsfs}
\usepackage{tikz}
\usepackage{subcaption}
\usetikzlibrary{arrows.meta,decorations.pathmorphing,calc,intersections}

\usepackage{qcircuit}
\usepackage{dsfont}
\usepackage{thm-restate}
\usepackage{longfbox}
\usepackage{comment}
\usepackage{algorithm}
\usepackage{algpseudocode}
\usepackage[style=alphabetic,minalphanames=3,maxalphanames=4,maxnames=99,backref=true]{biblatex}

\DeclareFieldFormat{eprint:iacr}{Cryptology ePrint Archive: \href{https://ia.cr/#1}{\texttt{#1}}}
\DeclareFieldFormat{eprint:iacrarchive}{Cryptology ePrint Archive: \href{https://eprint.iacr.org/archive/#1}{\texttt{#1}}}
\hypersetup{colorlinks={true},linkcolor={blue},citecolor=magenta}

\theoremstyle{plain}
\newtheorem{theorem}{Theorem}[section]
\newtheorem{lemma}[theorem]{Lemma}

\newtheorem*{theorem*}{Theorem}

\newtheorem{fact}[theorem]{Fact}
\newtheorem{corollary}[theorem]{Corollary}
\newtheorem{conjecture}{Conjecture}
\newtheorem{proposition}[theorem]{Proposition}
\newtheorem{definition}[theorem]{Definition}

\theoremstyle{remark}
\newtheorem{remark}[theorem]{Remark}

\newcommand{\R}{\mathbb{R}}

\newcommand{\N}{\mathbb{N}}

\renewcommand{\vec}[1]{\mathbf{#1}}
\newcommand{\reg}[1]{\mathsf{#1}}

\newcommand{\poly}{\mathrm{poly}}

\newcommand{\negl}{\mathrm{negl}}

\newcommand*{\interact}{\mathord{\leftrightarrows}}

\newcommand{\Sim}{\mathrm{Sim}}
\newcommand{\secparam}{\lambda}
\newcommand{\sk}{\mathrm{sk}}
\newcommand{\pp}{\mathrm{pp}}
\newcommand{\Enc}{\mathrm{Enc}}
\newcommand{\Dec}{\mathrm{Dec}}
\newcommand{\Gen}{\mathrm{Gen}}
\newcommand{\Com}{\mathrm{Com}}
\newcommand{\Setup}{\mathrm{Setup}}
\newcommand{\View}{\mathrm{View}}
\renewcommand{\vec}[1]{\mathbf{#1}}
\newcommand{\Commit}{\mathrm{Commit}}
\newcommand{\Reveal}{\mathrm{Reveal}}

\newcommand{\greg}[1]{\textcolor{blue}{(\bf Greg: #1)}}

\newcommand{\pos}{\textsf{loc}}

\Crefname{claim}{Claim}{Claims}
\Crefname{protocol}{Protocol}{Protocols}
\Crefname{protocolx}{Protocol}{Protocols}
\crefname{step}{Step}{Steps}
\Crefname{algorithm}{Protocol}{Protocols}

\floatname{algorithm}{Protocol}

\title{Private proofs of when and where}

\author{
\begin{tabular}{cccc}
    Uma Girish & Greg Gluch & Shafi Goldwasser & Tal Malkin\\
    \small Columbia & \small UC Berkeley & \small UC Berkeley & \small Columbia \\ 
    \\
     & Leo Orshansky & Henry Yuen & \\
     & \small Columbia & \small Columbia &
\end{tabular}
}

\begin{document}

\date{}

\maketitle
\begin{abstract}
Position verification schemes are interactive protocols where entities prove their physical location to others; this enables interactive proofs for statements of the form ``I am at a location $L$.'' Although secure position verification cannot be achieved with classical protocols (even with computational assumptions), they are feasible with \emph{quantum} protocols.

In this paper we introduce the notion of 
\emph{zero-knowledge position verification}, 
which extends position verification in two ways: 
\begin{enumerate}
    \item enabling entities to prove more sophisticated statements about their locations at different times (for example, proving statements about one's position in the {\it past} ``I was \emph{not} near location $L$ at noon yesterday'').  
    \item maintaining privacy for any other detail about their true location besides the statement they are proving.
\end{enumerate}
We construct zero-knowledge position verification 
from secure position verification and post-quantum one-way functions. 
The central tool in our construction is a primitive we call \emph{position commitment}, which allows entities to privately \emph{commit} to their physical position in a particular moment, which is then revealed at some later time. 

\end{abstract}

\section{Introduction}
\label{sec:intro}


Position verification protocols enable an entity (which we call the \emph{prover}) to prove the statement ``I am at location $L$'' to a set of verifiers. These protocols are analogous to interactive proofs in complexity theory and cryptography, except the statements proven are about a \emph{physical property} of the prover, rather than the logical properties of some mathematical object (e.g., ``The graph $\varphi$ is $3$-colorable''). Position verification gives rise to intriguing new cryptographic possibilities, such as using someone's geospatial location as a secure credential for decrypting a message or performing a private computation~\cite{chandran2009position,buhrman2014position}. 

The idea behind position verification is that a group of verifiers can exchange signals with a purported prover in a purported location, and based on the timing of the received signals, obtain trustworthy proof of the prover's location. Without strong assumptions, such proofs cannot be obtained in a purely classical setting: a number of papers~\cite{chandran2009position,kent2011quantum,buhrman2014position} showed that it is always possible for colluding adversaries to efficiently spoof any classical position verification protocol. This classical impossibility result holds even if we make computational assumptions; thus for position verification to be possible one has to consider different models, such as trusted hardware models~\cite{gabber1998prove},  bounded storage models~\cite{chandran2009position}, or the quantum model~\cite{kent2011quantum,buhrman2014position,malaney2016quantum}. The quantum position verification (QPV) model is perhaps the most compelling, as it leverages physical principles such as unclonability and uncertainty to potentially achieve secure location verification.

However, even if secure position verification is achievable, it raises a fundamental privacy concern: an entity's position in space and time is one of its most sensitive physical attributes. In many applications, the goal is not to reveal a precise location, but rather to prove some statement about it. For example, one might wish to prove to a judge that one was \emph{not} near a crime scene (``I was not at location $L$.''), or to certify that a nation-state's armaments remained within its borders (``The weapon was always inside a region.'').


In this work we ask:
\begin{center}
\textit{Can position verification protocols be made private?}
\end{center} 
We thus ask whether there exist proofs that \emph{don't} reveal anything additional about the locations of the involved entities (other than what is required by the application). In other words, is it possible to obtain \emph{zero-knowledge position proofs}? This would be analogous to traditional zero-knowledge proofs in complexity theory and cryptography.

We argue that many natural settings would benefit from zero-knowledge location-proving protocols. A striking example is a recent investigation \cite{strava2025} which revealed that presidential bodyguards inadvertently exposed the movements of heads of state by using the fitness-tracking app Strava.  
A protocol capable of proving statements such as ``Alice completed a run with a particular \emph{shape} within a given time,'' \emph{while revealing nothing about the run's geographic location}, would directly mitigate the vulnerability highlighted in \cite{strava2025}.  
Other potential applications include \emph{privacy-preserving verification of compliance with nuclear treaties}---for instance, the Antarctic Treaty \cite{antarctic_treaty_1959}, which forbids the placement of nuclear weapons in Antarctica, could benefit from proofs that a given device is not present on the continent.  
Further examples are \emph{private alibis} (proving that someone was not at a crime scene without revealing their actual location) and \emph{privacy-preserving export-control enforcement} (e.g., proving that a GPU is operated within a U.S.\ data center).

\paragraph{Contributions.}
Our contributions are:
\begin{enumerate}
    \item We formalize zero-knowledge position verification—more precisely, spacetime verification (see Section~\ref{sec:PVforspacetime})—using a simulation-based security definition.

    \item We construct zero-knowledge position verification protocols assuming the existence of post-quantum one-way functions and a certain class of position verification protocols, which is achieved by previous work (see Section~\ref{sec:zklp-intro}).

    \item We introduce a new primitive, \emph{position commitment}, which 
    allows a prover to commit to its geographic location $L$ at time $t$, and to open the 
    commitment at a later time. 
    We give a construction of position commitment 
    and show how it serves as a key building block for zero-knowledge position verification.
\end{enumerate}
We note that, beyond providing privacy, our protocols also enable the proving of statements about one's location at an arbitrary point in the \emph{past}. 
This could potentially allow, for example, someone to prove their past location to a judge.





\subsection{Position verification for regions of spacetime}\label{sec:PVforspacetime}

Before explaining zero-knowledge position verification, we first define the notion of position verification more carefully. 

We model space as $\R^d$ for some dimension $d$, and time is determined by a universal clock $t \in \mathbb{R}$. We assume that all parties know the time and their positions exactly. We assume that all signals between parties travel at some fixed velocity (e.g., one unit of distance per unit of time) to model the speed of light constraint in physics. We assume that parties can choose for their communication to be either broadcast or sent directionally, i.e., each party can transmit a signal along a ray. Finally, we assume that the behavior of all parties can be described by quantum mechanics -- that is, they can perform quantum computation and send quantum information to each other. (We elaborate more on the modeling in \Cref{sec:model}).

A position verification protocol is a triple $\Pi = (P,V,R)$ where $V = \{V_1,\ldots,V_k\}$ is a set of \emph{verifier} algorithms, $R \subset \R^d \times \R$ is a region of \emph{allowable prover points} in spacetime, and $P = \{P_\alpha\}_{\alpha \in R}$ is a set of \emph{prover} algorithms indexed by $R$. The protocol $\Pi$ is \emph{complete} if for all $\alpha = (L,t) \in R$, the prover $P_\alpha$ is positioned at $L$ at time $t$, and causes $V$ to accept with high probability. The protocol has \emph{position security against a class $\mathscr{C}$ of spoofers} if all sets $\{P_1,\ldots,P_m\}$ of \emph{spoofers} in the class $\mathscr{C}$, none of whom have positions in the set $R$, 
are rejected with high probability. (The formal definition can be found in \Cref{def:pv}). 


The security of a position verification protocol crucially depends on a class $\mathscr{C}$ of spoofers. 
It is known
that quantum position verification schemes cannot be secure against spoofers that use more than exponential amounts of quantum entanglement, because of the existence of a general attack~\cite{vaidman2003instantaneous,beigi2011simplified,buhrman2014position}. However, there are secure position verification protocols against spoofers with bounded entanglement~\cite{bluhm2022single} or bounded query complexity in the random oracle model~\cite{unruh2014quantum,liu2022beating}. For concreteness we will focus on spoofers with limited entanglement; however, our results generalize to broader classes of adversaries. 


In this definition, we think of the verifiers as being convinced of the statement ``There is a prover located in spacetime region $R$''. In the literature, position verification protocols have been typically defined for proving a specific position $L^*$ and time $t^*$. We call such protocols \emph{singleton position verification protocols}, because this corresponds to a singleton set $R = \{(L^*,t^*)\}$. Allowing a larger set of positions to be proved is important for our upcoming notion of zero-knowledge position verification.

\paragraph{A canonical position verification protocol.} The reader may find it useful to carry the following example of a quantum position verification protocol in mind. The protocol verifies that a prover is at the midpoint between two verifiers on a line; this is called the \emph{$f$-BB84} protocol~\cite{bluhm2022single}. Let $f:\{0,1\}^n \times \{0,1\}^n \to \{0,1\}$ be a boolean function. Using shared randomness, verifiers $V_1,V_2$ sample uniformly random strings $x,y \in \{0,1\}^n$ and a uniformly random bit $b \in \{0,1\}$. The first verifier $V_1$ creates a qubit in the state $\ket{\psi} = H^{f(x,y)} \ket{b}$ (i.e., a random BB84 state in a basis determined by $f(x,y)$). The verifier $V_1$ sends $x$ and $\ket{\psi}$ towards the midpoint and $V_2$ sends $y$ towards the midpoint (see \Cref{fig:1d-pos-ver} for an illustration). Each verifier expects a single bit $b'$ back from the purported prover and accept if $b = b'$.\footnote{The honest prover at the midpoint who receives $x,y,\ket{\psi}$ at the same time is supposed to compute $f(x,y)$ (which we assume can be done instantaneously), and measure the qubit $\ket{\psi}$ in the appropriate basis to recover the bit $b$.} When $f(x,y)$ is the inner product function this protocol is secure against spoofers that share $O(\log n)$ qubits of entanglement~\cite{bluhm2022single}, and it is conjectured that there exists a polynomial-time computable $f$ such that $f$-BB84 is secure against spoofers with superpolynomial qubits of entanglement.

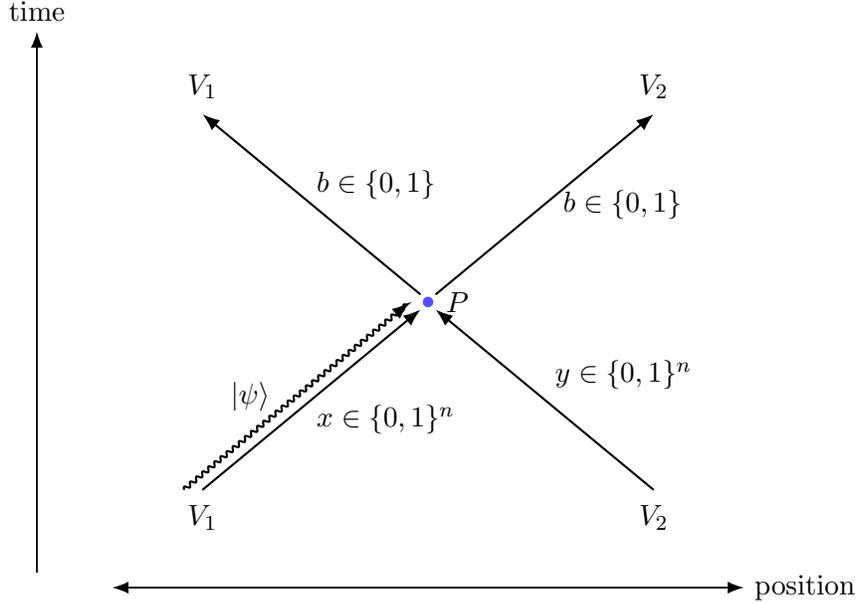
\begin{figure}[tbp!]
\begin{center}

\begin{tikzpicture}[x=1cm,y=1cm,>=Latex]
  \draw[-{Latex},line width=0.7pt] (-5.2,-3.6) -- (-5.2,3.6) node[above] {time};
  \draw[<->,line width=0.7pt] (-4.2,-3.8) -- (4.2,-3.8) node[right] {position};

  \coordinate (V0bot) at (-3.0,-2.5); 
  \coordinate (V1bot) at ( 3.0,-2.5); 
  \coordinate (P)     at ( 0.0, 0.0); 
  \coordinate (V0top) at (-3.0, 2.5); 
  \coordinate (V1top) at ( 3.0, 2.5); 

  \node[below=2pt] at (V0bot) {$V_1$};
  \node[below=2pt] at (V1bot) {$V_2$};
  \fill[blue!70] (P) circle (0.07);
  \node[right=3pt] at (P) {$P$};
  \node[above=2pt] at (V0top) {$V_1$};
  \node[above=2pt] at (V1top) {$V_2$};

  \draw[->,thick] (V0bot) -- ($(P)+(-0.1,-0.1)$)
    node[midway,below right = -2pt] {$x \in \{0,1\}^n$};

  \draw[decorate,decoration={snake,segment length=3pt,amplitude=0.8pt},->,thick]
    ($(V0bot)+(-0.25,0)$) -- ($(P)+(-0.25,-0.01)$)
    node[midway,left=6pt] {$\ket{\psi}$};

  \draw[->,thick] (V1bot) -- ($(P)+(0.1,-0.1)$)
    node[midway,above right] {$y \in \{0,1\}^n$};

  \draw[->,thick] ($(P)+(-0.1,0.1)$) -- (V0top)
    node[midway,above right = -2pt] {$b \in \{0,1\}$};
  \draw[->,thick] ($(P)+(0.1,0.1)$) -- (V1top)
    node[midway,right = 3pt] {$b \in \{0,1\}$};
\end{tikzpicture}

\end{center}
    \caption{A spacetime diagram of the $f$-BB84 protocol~\cite{bluhm2022single}, specialized to 1 dimension position verification. Time goes up, position is horizontal, and signals travel along 45$^{\circ}$ angles.}
    \label{fig:1d-pos-ver}
\end{figure} 

\subsection{Zero-knowledge position verification}\label{sec:zklp-intro}


 We now give an informal definition of zero-knowledge position verification. Intuitively, the privacy notion we want to capture is that  the ``view'' of the verifiers when interacting with a prover $P$ located in the allowed region $R$ can be simulated independently of its exact position within $R$. Let $\Pi = (P,V,R)$ be a position verification protocol with honest prover $P$, honest verifiers $V = \{V_1,\ldots,V_k\}$, and a set of allowable prover spacetime points $R \subseteq \R^d \times \R$. We say that $\Pi$ is \emph{zero-knowledge} if there exists a polynomial-time \emph{simulator} outputting a quantum state $\rho$ such that for all allowable spacetime points $\alpha \in R$, $\rho$ is computationally indistinguishable from the joint state $\rho_\alpha$ of verifiers $V$ at any given moment during the interaction with prover $P$ located at $\alpha$.  
 

Just like there are many variants of zero-knowledge in cryptography (statistical vs.\ computational, semi-honest vs.\ malicious verifier, etc.), there can also be many variants of zero-knowledge position verification. In this paper we study the setting of \emph{semi-honest} verifiers, where we assume that the verifiers honestly follow the protocol $\Pi$, but may try to learn more about the prover's true location by inspecting the prover's messages. We discuss ideas and barriers towards relaxing the semi-honest verifier model in \Cref{sec:towards-malicious}.



Next we show that zero-knowledge position verification is achievable assuming the existence of one-way functions (OWF) and secure singleton position verification protocols satisfying the following conditions:
\begin{enumerate}
    \item The protocol is one round: the verifiers send a message to the prover, who is supposed to respond immediately. 
    \item The prover's message is classical, and the verifiers' accept/reject decision is a deterministic function of the prover's response, timing, and the verifier's shared randomness. 
    \item The protocol can securely verify any position contained within the convex hull of the verifiers' spatial positions.
\end{enumerate}
We call a singleton position verification protocol satisfying this \emph{nice}; see \Cref{def:one-shot-mostly-classical} for a formal definition, and an intuition for these criteria.
An example of a nice singleton position verification scheme is $f$-BB84. Other examples include the classical-verifier protocol of \cite{liu2022beating}, or the random-oracle model protocol of \cite{unruh2014quantum}.

We now can state the main result of the paper.

\begin{theorem}[Zero-knowledge position verification, informal]
\label{thm:zklp-intro}
    Suppose there exist post-quantum one-way functions and nice singleton position verification protocols which have position security against spoofers sharing at most $E$ entangled qubits. Then for all finite $R \subseteq \R^d \times \R$, there is an (honest-verifier) zero-knowledge position verification protocol $\Pi_{ZK} = (P,V,R)$ with position security against spoofers sharing $E/2$ entangled qubits.
\end{theorem}

In other words, any nice secure position verification protocol $\Pi$ 
for proving \emph{single} spacetime point can be ``upgraded'' (assuming 
OWF) to a \emph{zero-knowledge} position verification protocol $\Pi_{ZK}$ for
any \emph{finite set} of spacetime positions. 

A formal version of this main result is stated and proven in \Cref{thm:main}.

\begin{remark}
    As a default choice we assume the quantum model of position verification, where the best known protocols are secure against adversaries with unlimited computational power, but with a bounded amount of shared entanglement. However, since the position security of the zero-knowledge location proofs construction is directly inherited from the underlying singleton position verification protocol (plus security of the OWF against the verifiers)\footnote{What our construction actually uses is secret-key encryption and bit commitment, which are equivalent to OWF.},
    our results hold in more general settings. For example, classical position verification is trivially possible if you assume there is a single adversarial spoofer; timing the messages of the prover will pinpoint its location. Plugging this simple positioning protocol into \Cref{thm:zklp-intro} yields an interesting example of a \emph{purely classical} zero-knowledge location protocol secure against single-prover adversaries, assuming standard OWF exist. Similarly, our results also yield zero-knowledge location protocols in the classical-verifier, quantum-prover model of~\cite{liu2022beating}, where security is computational (see \Cref{sec:optimization} for an in-depth explanation).
\end{remark}

\subsection{Technical overview}
\label{sec:technical-overview}

The key technical tool in our construction of zero-knowledge position verification is an object we call \emph{position commitment}. This is a protocol that allows a prover to commit to its spacetime position $(L,t)$ without initially revealing anything about it to the verifiers (i.e., the position commitment is \emph{hiding}). Some time later, the prover can decide to \emph{reveal} the pair $(L,t)$ to the verifiers, who can perform a check to ensure that the prover was truly in the location $L$ at time $t$ (i.e., the position commitment is \emph{binding}).\footnote{We call this second property ``position binding'' to mirror the binding property of cryptographic commitments. Note, however, that our notion is significantly stronger than conventional binding: the prover is not only bound to the choice of a single point $(L,t)$ -- it must have \emph{physically occupied} $(L,t)$.}

We define and construct this primitive primarily as a building block for zero-knowledge position proofs, but we point out that position commitments by themselves already enable interesting applications. For example, an individual could continuously commit to their position at regular intervals throughout the day, which would establish a \emph{private} record of their whereabouts. Later, if necessary, they could unseal part of this record in a courtroom, to convince a judge of their alibi.

Our first technical result is that position commitments can be constructed from one-way functions and nice singleton position verification protocols.
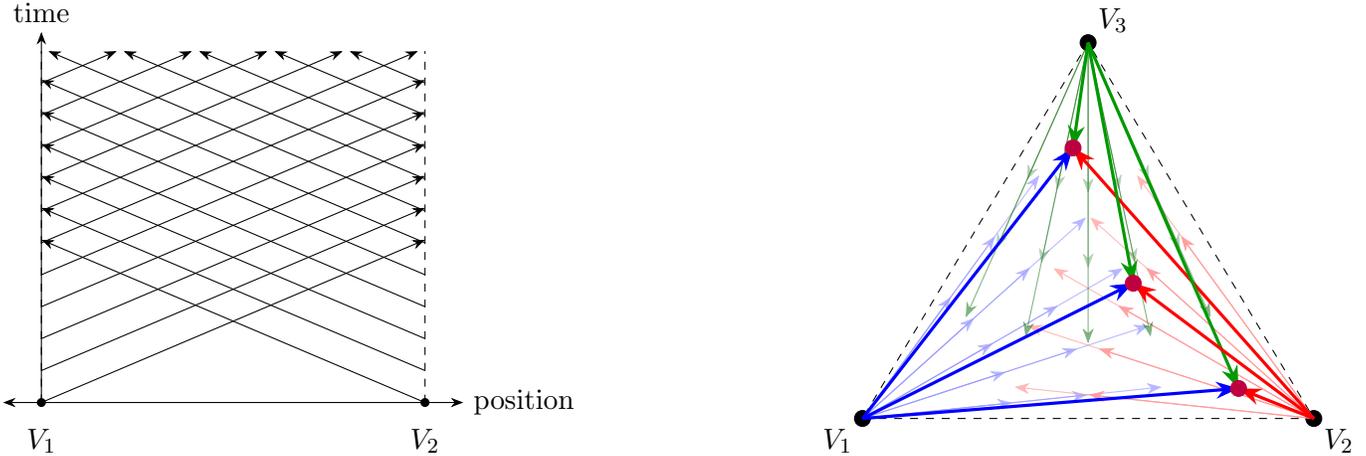
\begin{figure}[ht]
    \centering
    \hspace{-1cm}
    \begin{subfigure}{0.4\textwidth}
        \centering
        \begin{tikzpicture}[scale=0.85, >=Stealth]
          \def\nsteps{10}
          \def\L{6}
          \def\dy{0.5}
          \def\ang{23}
          \pgfmathsetmacro{\ymax}{(\nsteps+1)*\dy}
        
          \draw[->] (0,0) -- (0,{\ymax+0.3}) node[above] {time};
          \draw[<->] (-0.6,0) -- ({\L+0.6},0) node[right] {position};
        
          \draw[dashed] (0,0) -- (0,\ymax);
          \draw[dashed] (\L,0) -- (\L,\ymax);
        
          \fill (0,0) circle (2pt);
          \fill (\L,0) circle (2pt);
          \node[below] at (0,-0.25) {$V_1$};
          \node[below] at (\L,-0.25) {$V_2$};
        
          \foreach \i in {0,...,\nsteps} {
            \pgfmathsetmacro{\y}{\i*\dy}
            \pgfmathsetmacro{\cosang}{cos(\ang)}
            \pgfmathsetmacro{\sinang}{sin(\ang)}
            \pgfmathsetmacro{\sx}{\L/\cosang}
            \pgfmathsetmacro{\sy}{(\ymax-\y)/\sinang}
            \pgfmathsetmacro{\s}{min(\sx,\sy)}
        
            \pgfmathsetmacro{\dx}{\s*\cosang}
            \pgfmathsetmacro{\dyend}{\s*\sinang}
            \pgfmathsetmacro{\xendL}{\dx}
            \pgfmathsetmacro{\yendL}{\y+\dyend}
            \draw[->] (0,\y) -- (\xendL,\yendL);
        
            \pgfmathsetmacro{\xendR}{\L-\dx}
            \pgfmathsetmacro{\yendR}{\y+\dyend}
            \draw[->] (\L,\y) -- (\xendR,\yendR);
          }
        
        \end{tikzpicture}
    \end{subfigure}%
    \hfill
    \begin{subfigure}{0.4\textwidth}
        \centering
        \begin{tikzpicture}[scale=2]
            \coordinate (A) at (0,0);
            \coordinate (B) at (3,0);
            \coordinate (C) at (1.5, 2.5);
            \coordinate (P1) at (1.4, 1.8); 
            \coordinate (P2) at (2.5, 0.2); 
            \coordinate (P3) at (1.8, 0.9); 
    
            \node[below left] at (A) {$V_1$};
            \node[below right] at (B) {$V_2$};
            \node[above right] at (C) {$V_3$};
            \filldraw[black] (A) circle (1.5pt);
            \filldraw[black] (B) circle (1.5pt);
            \filldraw[black] (C) circle (1.5pt);
    
            \draw[black, dashed] (A) -- (B) -- (C) -- cycle;
    
            \foreach \angle in {6,18,...,54} {
                \draw[-{Stealth[length=2mm]}, blue!40, opacity=0.7] (A) -- ++(\angle:2);
                \draw[-{Stealth[length=2mm]}, blue!40, opacity=0.7] (A) -- ++(\angle:1.5);
                \draw[-{Stealth[length=2mm]}, blue!40, opacity=0.7] (A) -- ++(\angle:1);
            }
            \foreach \angle in {126,138,...,174} {
                \draw[-{Stealth[length=2mm]}, red!40, opacity=0.7] (B) -- ++(\angle:2);
                \draw[-{Stealth[length=2mm]}, red!40, opacity=0.7] (B) -- ++(\angle:1.5);
                \draw[-{Stealth[length=2mm]}, red!40, opacity=0.7] (B) -- ++(\angle:1.5);
            }
            \foreach \angle in {-66,-78,...,-114} {
                \draw[-{Stealth[length=2mm]}, green!40!black, opacity=0.5] (C) -- ++(\angle:2);
                \draw[-{Stealth[length=2mm]}, green!40!black, opacity=0.3] (C) -- ++(\angle:1.5);
                \draw[-{Stealth[length=2mm]}, green!40!black, opacity=0.3] (C) -- ++(\angle:1);
            }
    
            \draw[-{Stealth[length=3mm]}, blue, very thick] (A) -- (P1);
            \draw[-{Stealth[length=3mm]}, red, very thick] (B) -- (P1);
            \draw[-{Stealth[length=3mm]}, green!60!black, very thick] (C) -- (P1);
            
            \draw[-{Stealth[length=3mm]}, blue, very thick] (A) -- (P2);
            \draw[-{Stealth[length=3mm]}, red, very thick] (B) -- (P2);
            \draw[-{Stealth[length=3mm]}, green!60!black, very thick] (C) -- (P2);
    
            \draw[-{Stealth[length=3mm]}, blue, very thick] (A) -- (P3);
            \draw[-{Stealth[length=3mm]}, red, very thick] (B) -- (P3);
            \draw[-{Stealth[length=3mm]}, green!60!black, very thick] (C) -- (P3);
    
            \filldraw[purple] (P1) circle (1.5pt);
            \filldraw[purple] (P2) circle (1.5pt);
            \filldraw[purple] (P3) circle (1.5pt);
    
        \end{tikzpicture}
    \end{subfigure}
    \caption{The position commitment described in \Cref{def:pc-construction} pictured for 1D (left), and 2D (right). In $d$ spatial dimensions, intersections between $d+1$ different verifier messages represent spacetime points at which an honest prover might receive a set of position verification challenges, which they will respond to in an encrypted way. A few example points have been bolded on the right for clarity.}
    \label{fig:2d-pc}
\end{figure}
\begin{theorem}[informal version of \Cref{thm:construction-secure}]
\label{thm:pc-intro}
    Suppose there exist post-quantum one-way functions and nice singleton position verification protocols which have position security against spoofers sharing at most $E$ entangled qubits. Then for every finite spacetime region $S$, there exists a position commitment protocol with completeness for all points in $S$, computational hiding against semi-honest verifiers, and statistical position binding against spoofers sharing $E/2$ qubits of entanglement.
\end{theorem}

\paragraph{Constructing position commitments.} We give a high level overview of our construction of position commitments. We'll outline what the honest verifiers $V$ and honest prover $P$ do. 


Suppose prover $P$ is planning to be in the spacetime location $(L^*,t^*)\in S$ where $t^*$ is in the future. Without announcing $(L^*,t^*)$ to the verifiers, $P$ initiates a \emph{session} that takes place between times $t_{\mathrm{init}} < t_{\min}(S)< t_{\max}(S) < t_{\mathrm{final}}$ 
wherein it interacts with $V$ to commit to its hidden position (here, $t_{\min}(S)$ and $t_{\max}(S)$ are the earliest and latest times of points in $S$, respectively). In this session, the prover and the verifiers engage in a continuous stream of singleton position verification protocols, running enough instances to cover the entire space of committable points where $P$ may be located. However, in order not to expose its location, $P$ encrypts all of its messages under a secret key, which the verifiers will only learn when the prover decides it wants to reveal its commitment. The protocol is constructed such that no matter where $P$ is, it can convincingly pretend to be at all other points in $S$ simultaneously, by sending dummy messages which meet those timing constraints. In order to prevent the prover from changing its mind about the contents of the encrypted messages by finding a different decryption key, the prover commits to the encryption secret key at the beginning of the protocol. The protocol runs in time $\poly(\secparam,|S|)$. 
We illustrate the position commitment protocol for dimensions $d = 1$ and $d = 2$ in \Cref{fig:2d-pc}.

\paragraph{Proving security of position commitments.}
It is not difficult to see that \emph{hiding} security of the above construction follows from two facts: 1. all prover messages are indistinguishable assuming the security of secret-key encryption, and 2. the set of message timings received by the verifiers is kept static, regardless of the prover's true position.

Proving the \emph{position binding} of this construction is, however, considerably less straightforward. The proof, a reduction to position verification soundness and classical commitment binding, is unlike typical reductions in the cryptography literature. Namely, there is a challenge in manipulating the \emph{real-time} inputs and outputs of the adversary -- it is unclear what ``fast-forwarding'' and ``rewinding'' might mean in a context where the adversary is distributed across concrete points in space and time. In our solution to this, the reduction is required to take \emph{nonuniform quantum advice}.

\paragraph{From position commitments to zero-knowledge position verification.} 
We now describe how position commitments can be used as a building block to construct zero-knowledge position proofs, which proves \Cref{thm:zklp-intro}. 
The method we use for this extension is very similar to how, in the classical case, commitments are used to achieve zero-knowledge for general $\sf NP$ predicates. Once a prover has interacted with some verifiers to create the position commitment (which is, in our construction from nice position verification protocols, a classical string), it engages in a zero-knowledge proof protocol for some statement about the hidden position. In particular, the prover can express an $\sf NP$-verifiable statement of the fact that it possesses an opening to the commitment such that the revealed point belongs to $R$, and employ the standard zero-knowledge proof system for $\sf NP$ languages. The soundness of this proof, combined with statistical position binding of our position commitments, gives us position security against dishonest provers.

\paragraph{Optimizing the position commitment protocol.}
We construct position commitments for a fairly general class of position verification protocols (all so-called \emph{nice} protocols), which encompasses most of what has been explored in the literature. However, this generality comes with a cost to the efficiency of our protocols. 
In \Cref{sec:optimization}, we narrow our assumption further to protocols whose messages are fully classical, and whose verifier challenges can be generated independently per verifier, and show how this can be used to greatly optimize the computational requirements of our protocol.

\subsection{A brief overview of quantum position verification}

The study of quantum position verification (QPV) has blossomed in recent years. A number of protocols have been proposed~\cite{kent2011quantum,buhrman2014position,unruh2014quantum,bluhm2022single,liu2022beating,escolafarras2025quantum}, some with provable security guarantees, and others' that are conjectured. Fascinatingly, unconditionally-secure position verification is not possible even in the quantum setting; it was shown by~\cite{buhrman2014position} that the so-called ``instantaneous nonlocal quantum computation'' protocol of Vaidman~\cite{vaidman2003instantaneous} (subsequently improved by~\cite{beigi2011simplified}) can be used to break any QPV protocol, provided that the attackers can utilize exponential amounts of quantum entanglement. The central question in this area is to prove the security of a QPV protocol in the presence of adversaries that can only share polynomial amounts of entanglement. This challenge, in turn, is connected with deep questions in a number of seemingly-unrelated areas, ranging from \emph{classical} lower bounds for information-theoretic cryptography~\cite{allerstorfer2024relating,asadi2025conditional,asadi2025ranklowerboundsnonlocal,girish2025comparingclassicalquantumconditional}, the AdS/CFT correspondence in quantum gravity~\cite{may2019quantum}, and conjectures about Banach spaces~\cite{junge2022geometry}.

 A recent interesting result by \cite{jpm2024} showed that classical-verifier position verification protocols such as that of \cite{liu2022beating} are, in general, equivalent to protocols for \emph{certifiable randomness}, hinting at deep connections between QPV and other areas of quantum cryptography.

\subsection{Summary and future directions}


Our concept of zero-knowledge position verification gives a solution towards achieving both privacy and provability for statements about one's physical location. We show that it can be achieved using one-way functions and (non-private) quantum position verification protocols. 
We view our contribution as taking the first steps in exploring the natural combination of privacy and proofs of position. We identify some interesting directions for future work.





\paragraph{Towards malicious-verifier security (full discussion in \Cref{sec:towards-malicious}).} We define and construct position commitments and zero-knowledge position verification protocols where the privacy condition holds with respect to \emph{semi-honest} verifiers. We view this model as well-motivated by itself; for example, one can imagine that the verifiers are cell towers whose goal is to provide consistent position verification coverage to many customers over a wide area. Although the cell phone company may be interested in tracking their customers' whereabouts by examining communication logs, it is not sufficiently motivated to deviate from the protocol in a way that can be noticed by customers. 

Can we obtain a stronger form of zero-knowledge position verification against malicious verifiers who deviate from the protocol?  Unfortunately, we cannot hope to achieve zero-knowledge against arbitrary cheating verifiers; in \Cref{sec:towards-malicious} we present a generic attack where the verifiers can determine whether a prover is in a location $L$ by actively sending directional messages towards $L$, but nowhere else. Whether or not the prover responds will leak whether the prover is in location $L$. 
We believe this is an interesting property, which is inherent to any protocol where the honest prover would respond to directional messages (either because the honest verifier behavior requires sending directional messages, or because it cannot distinguish between messages that are directional vs broadcast). 
On the other hand, we conjecture that in other settings, a more standard ``HVZK-to-ZK'' upgrading techniques should apply, and security against malicious verifiers can be added.  That is: 
\begin{conjecture}[informal]
    Suppose we have a honest-verifier position verification protocol which requires the verifier to send only broadcast messages, and further suppose that the honest prover can tell whether a message it received was directional rather than broadcast.  Then the protocol can be upgraded to have security against malicious verifiers. 
 \end{conjecture} 
We note that position verification protocols that require only broadcast messages exist (e.g., ones where the communication is classical). Thus, we believe that formalizing our attack and proving the above conjecture is a promising avenue for future work.

\paragraph{Towards stronger position verification guarantees.} 
When we carefully consider the soundness of position verification, we find that there is a subtle gap between the \emph{intuitive} notion of position security and the \emph{formal} definition presented here and in previous papers on position verification, starting from the earliest ones~\cite{chandran2009position,buhrman2014position}. In the intuitive notion, a position verification protocol allows someone to prove the statement ``I was at location $X$''. However, in the formal definition, the protocol only guarantees that, of a coalition of colluding provers, \emph{at least one} of the provers (rather than \emph{the} prover) was at location $X$. 
It may be desirable to 
have a stronger guarantee that {all} of the ``real'' computation done to pass the protocol was done at the claimed location.

How do we define who the ``real" prover or computation is (who is ``I")?  
One potential direction is to leverage unclonability.
Imagine a position verification protocol with the following structure:
\begin{enumerate}
    \item At the beginning of time, the prover receives from the verifiers a quantum state $\ket{\psi}$, which is unclonable and has some cryptographic function. Let's say, for example, that $\ket{\psi}$ is a quantum signing token (like that of \cite{BS23}).
    \item During the position verification protocol, the prover is required to sign its responses with $\ket{\psi}$.
    \item By an appropriate security reduction, we would then be able to claim, ``the PV protocol passes if and only if $\ket{\psi}$ was in location $X$''.
\end{enumerate}
We note, however, that formalizing this security notion, with the goal of fully capturing our intuitive notion, still seems difficult. What does it mean for $\ket{\psi}$ to be \emph{in} location $X$ -- what if it is distributed across a large error-correcting code? Questions like these are explored in a 2016 paper \cite{Hayden_2016}, but it is not immediately clear how to apply their results to this setting. An additional question is, under what assumptions can we say that $\ket{\psi}$ truly refers to a \emph{unique} object in the universe, which has a canonical position that we can refer to? Is some version of this statement implied by unclonability alone?




\paragraph{Simplifying our nonuniform security reduction.} In \Cref{sec:position-commitments} we give a construction of a new position-based cryptographic primitive, which we call position commitments; we then prove the soundness of this construction by a reduction to secure position verification and classical commitment binding. The reduction, which is found in \Cref{thm:construction-secure}, operates in an atypical way and requires nonuniform quantum advice. The reasons for this are, in brief, as follows.

Our reduction has white-box access to a distributed cheating adversary for our position commitment scheme, and must use this to win in the standard position verification experiment (or break the statistical binding of Naor commitments). In the position commitment experiment, the cheating adversary receives a set of PV challenges, and outputs a set of \emph{encrypted} PV responses, which will then be decrypted at a later time whenever the adversary chooses to reveal the secret key. A natural reduction strategy would be to feed this adversary the PV challenges, and quickly fast-forward the adversary until it 1. outputs the encrypted responses, and then later 2. outputs the decryption key -- after this, the reduction can simply output the plaintext PV responses and win in the PV experiment.

However, we run into a problem: the distributed adversary might need to exchange messages between its various components, \emph{which are spatially distant}. Therefore, it is physically impossible for our reduction to fast-forward a general adversary by any significant amount of time, as the reduction needs to respond instantaneously to its challenges and cannot wait for the exchanging of messages at the speed of light. In our proof of \Cref{thm:construction-secure} by endowing the reduction with a nonuniform quantum advice state, containing the secret decryption key that the adversary will eventually output. Can this be simplified, and the nonuniformity removed?

\paragraph{Towards more realistic models.} Finally, 
the model we used in this paper (and which is consistent with prior work) is highly idealized, including precise, synchronized clocks, and instantaneous computation. Some of this seems inherent not only to the main ideas  underlying existing protocols (based on timing of responses) but also to any real-world quantum information processing. Finding  a way to relax these assumptions could help bring these protocols a step closer to practical realizability.


\section{Preliminaries}\label{sec:preliminaries}
\subsection{Notation}
We take $\secparam$ to be the security parameter. We write PPT and QPT to denote probabilistic polynomial-time and quantum polynomial-time, respectively. We write $\negl(\secparam)$ to denote a negligible function in $\secparam$. For two algorithms $P$ and $V$ (implicitly, prover and verifier), we write $(P\interact V)\in\{\mathrm{accept},\mathrm{reject}\}$ to denote the interaction between $P$ and $V$, where the accept/reject decision is output by $V$ at the conclusion of the interaction. When the parties share a common input $x$, the interaction will be denoted $(P\interact V)(x)$. $V$'s view of a (classical-message) interaction with $P$, denoted $\View_V(P\interact V)$, is defined as the full transcript of messages exchanged between $P$ and $V$, as well as all random coins flipped by $V$, during the interaction. For quantum-message interactions between $P$ and $V$, we define $V$'s view at a given time $\tau$, denoted $\View^\tau_V(P\interact V)$, to be the quantum mixed state corresponding to $V$'s internal state at time $\tau$. 

\subsection{The model}
\label{sec:model}

We now describe the model of space, time, and interaction used in the paper.

\paragraph{Spacetime.} We model space as $\R^d$ for some dimension $d$, and time as $\R$. We assume that all parties have access to a synchronized clock that reports the current time $t \in \R$. A \emph{point in spacetime} is a pair $(L,t) \in \R^d\times\R$, combining a spatial point with a time. For the sake of simplicity, we assume that all parties are located at points $(L,t) \in \mathbb{Q}^d\times\mathbb{Q}$; this allows finite representations of their positions.

Each algorithm $A$ is identified with a unique spatial point at all times $t$, called its \emph{position}. We use the following notations interchangeably:
\begin{enumerate}
    \item $A$ occupies/is positioned at the (spacetime) point $(L,t)$.
    \item $A$'s position at time $t$ is $L$.
    \item $L=\pos_t(A)$
\end{enumerate}
For a spacetime region $R\subseteq\R^d\times\R$, we may write $\pos(A)\in R$ to mean that there exists a time $t$ such that $(\pos_t(A),t)\in R$, and $\pos(A)\not\in R$ for the negation of this.


\paragraph{Signals.} We assume there are two types of signals: \emph{directional} and \emph{broadcast}. A directional signal is emitted from some location and travels along some ray in $\R^d$, and only spatial points located on the ray can receive the signal (e.g., a party can transmit a ``laser beam''). A broadcast signal is emitted from some location and eventually can be received by all points in space (e.g., a party can send a ``radio broadcast''). In general, a signal is a quantum state, although we often distinguish between the classical and quantum parts of a given signal (for example, a party may send a string $x$ and a qubit $\ket{\psi}$). Only classical strings will be broadcast.

By convention, we assume that signals travel at a fixed speed of one unit of space per unit time (i.e., the ``speed of light''). For example, if a directional signal is sent from the origin in the direction of some unit vector $\vec{m} \in \R^d$ at time $0$, the signal arrives at point $t \vec{m} \in \R^d$ at every subsequent time $t > 0$. For broadcast signals, the previous statement is true for \emph{all} unit vectors $\vec{m}\in\R^d$.



\paragraph{Computation and Interaction.} 
We model each party as some quantum algorithm $A$ that can receive and transmit classical and quantum information. Its classical inputs are: a security parameter $1^\secparam$ represented in unary, its current location $L$ (expressed in binary), and the time $t$ (expressed in binary). Its quantum inputs are a register containing its internal memory state, and a register representing the (quantum) messages that are received at precisely time $t$. If there are no messages received, then the message register is given some encoding of the $\bot$ symbol. The output of the algorithm $A$ is a pair $(\mathtt{mode},\vec{m})$ and two quantum registers where
\begin{enumerate}
    \item $\mathtt{mode} \in \{ \mathtt{directional}, \mathtt{broadcast}, \bot \}$ indicates whether the message is supposed to be directional, broadcast, or not sent at all,
    \item The vector $\vec{m} \in \R^d$ specifies the direction if needed, and
    \item The first quantum register is the party's internal memory state (to be passed to the same party in the next time step), and the second quantum register contains the message to be sent.
\end{enumerate} 
If $\mathtt{mode} = \mathtt{broadcast}$, then the state $\rho$ is measured to obtain a classical string $x$ and broadcast everywhere. 

We assume that algorithms have two interfaces with the global clock: a discrete ``tick'' update on which they can run their loops, and a ``real-time'' interface with which they can, for example, send out a message at any $t\in\R$. Since we assume that each party has a unique position $L \in \R^d$ at each time $t \in \R $, we model its behavior at each clock tick $t$ as computing the output $A(1^\secparam,L,t,\reg{E},\reg{A})$ where $\reg{E},\reg{A}$ denote its internal memory and message registers, respectively; the message register is some well-defined concatenation of all the signals received at spacetime point $(L,t)$. We will henceforth elide these inputs (save for the security parameter) as they are implicit from context.

By default, we assume that the computational complexity of the quantum algorithm $A$ is polynomial in the security parameter $\secparam$. For simplicity, we separate the time it takes to compute $A(1^\secparam,L,t,\cdots)$ from the ticking of the global clock $t$. That is, we assume $A$ computes its output in one time step.


In this paper, there are two types of parties: verifiers and provers. We focus on honest verifiers whose locations are always known to all parties. We discuss possible approaches to malicious verifiers in  \Cref{sec:towards-malicious}.

These modeling decisions were chosen to be consistent with the prior literature on position verification, while being as simple as possible in order to illustrate the essential ideas behind zero knowledge position verification. Extending the model to handle inaccuracies, errors, unsynchronized clocks, etc. are fascinating directions for future work.

\subsection{Cryptographic Notions}

We recall the definitions of some basic cryptographic primitives that will be used in the paper.

\begin{definition}[Computational Indistinguishability]
    The families of random variables $\{X_\lambda\}_{\lambda\in\mathbb{N}}$ and $\{Y_\lambda\}_{\lambda\in\mathbb{N}}$ are \textbf{computationally indistinguishable} if for every QPT distinguisher $D$, there exists a negligible function $\negl$ such that for all $\secparam\in\N$:
    \[
\left|\Pr\left[D(1^\lambda,X_\lambda)=1\right]-\Pr\left[D(1^\lambda,Y_\lambda)=1\right]\right|\leq\negl(\lambda)~.
\]
For random variables $X$ and $Y$, where the dependence on $\lambda$ is implicit from context, we denote the above with $X\approx_c Y$. Similarly, we say that families of quantum states $\{ \rho_\secparam \}_{\secparam \in \N}$ and $\{ \sigma_\secparam \}_{\secparam \in \N}$ are computationally indistinguishable if they cannot be distinguished by QPT algorithms with more than negligible advantage.
\end{definition}


\begin{definition}[Secret-Key Encryption]
\label{def:secret-key-enc}
    Let $\mathcal{E}=(\Gen,\Enc,\Dec)$ be a tuple of PPT algorithms, such that for all $\secparam \in \N$ and all $x\in\{0,1\}^*$, 
    \[
    \Pr\left[\Dec(\sk,\Enc(\sk,x))=x : \sk \gets{\Gen(1^\secparam)}\right]=1~.
    \]
    We say $\mathcal{E}$ is a \textbf{post-quantum encryption scheme} if for any QPT algorithm $A$ with oracle access to $\Enc(\sk,\cdot)$ and which outputs two equal-size lists of plaintexts, the following holds. With $(a_1\dots a_n),(b_1\dots b_n)\gets A^{\Enc(\sk,\cdot)}(1^\secparam)$ and $\sk\gets\Gen(1^\secparam)$, 
    \[(\Enc(\sk,a_1)\dots\Enc(\sk,a_n))\approx_c(\Enc(\sk,b_1)\dots\Enc(\sk,b_n))~.\]
\end{definition}
\begin{definition}[Commitment Scheme]
    A \textbf{commitment scheme} $\mathcal{C}$ is a pair $(\Setup,\Com)$ of PPT algorithms with the following syntax: 
    \begin{itemize}
        \item $\Setup(1^\secparam)$ outputs a random string $\pp$ (for ``public parameters''), which we assume contains the security parameter $1^\secparam$.
        \item $\Com(\pp, x ; r)$ is a deterministic function of $(\pp,x,r)$ where $r \in \{0,1\}^\secparam$.
    \end{itemize}
    We say that $\mathcal{C}$ is \textbf{post-quantum computationally hiding} if for all QPT algorithms $A$, with $\pp\gets\Setup(1^\secparam)$ and $(x_0,x_1)\gets A(1^\secparam,\pp)$, the following statement holds:
    \[
        \left \{ \Com(\pp,x_0; r) : r \gets \{0,1\}^\secparam \right \}_{\secparam \in \N} \approx_c \left \{ \Com(\pp,x_1; r) : r \gets \{0,1\}^\secparam \right \}_{\secparam \in \N}~.
    \]    
    We say that $\mathcal{C}$ is \textbf{statistically binding} if for all (computationally unbounded) algorithms $B$, there is a negligible function $\negl$ such that the following holds for all $\secparam\in\N$: 
    \[\Pr\left[
    \begin{array}{c}
        x_1\neq x_2 \\
        \Com(\pp,x_1 ; r_1)=\Com(\pp,x_2 ; r_2)
    \end{array}
    :
    \begin{array}{c}
        \pp\gets \Setup(1^\lambda) \\
        (r_1,r_2,x_1, x_2) \gets B(\pp)
    \end{array}
    \right]\leq\negl(\lambda)~.\]
\end{definition}
\begin{definition}[Zero-Knowledge Proof]
Let $L \subseteq \{0,1\}^*$ be a language. A pair $(P,V)$ of classical PPT interactive algorithms is an \textbf{honest-verifier, (post-quantum) computational zero-knowledge proof system} for $L$ with completeness error $c(\secparam)$, soundness error $s(\secparam)$ if the following hold.
\begin{enumerate}
    \item \textbf{(Completeness)} For all $\secparam \in \N$ and for all $x\in L$, there exists a string $w$ (called the witness) such that
    \[
\Pr \left [ (P(w) \interact  V)(1^\secparam,x) \text{ accepts} \right] \geq 1 - c(\secparam)
    \]
    \item \textbf{(Soundness)} For all $\secparam \in \N$ and $x\not\in L$, for all provers $P^*$, 
    \[
        \Pr \left [(P^*\interact  V)(1^\secparam,x)  \text{ accepts} \right] \leq s(\secparam)~.
    \]    
    \item \textbf{(Honest-Verifier Post-Quantum Computational Zero Knowledge)} There exists a PPT simulator $\Sim$ such that for all $x\in L$, the following holds: 
    \[ \Big \{ \Sim(1^\secparam,x) \Big \}_{\secparam \in \N} \approx_c \Big \{ \View_V(P(w) \interact  V)(1^\secparam,x) \Big \}_{\secparam \in \N} \]
    for all witnesses $w$ that make the completeness condition hold.
\end{enumerate}
\end{definition}

\begin{fact}[Existence of post-quantum encryption, commitments, and zero-knowledge proofs]

Assuming the existence of post-quantum one-way functions, there exist the following:
\begin{enumerate}
    \item Post-quantum secret-key encryption~\cite{katzlindell},
    \item Post-quantum computationally hiding and statistically binding commitments~\cite{naor91}, and
    \item Post-quantum honest-verifier computational zero-knowledge proof systems for any $L \in \mathsf{NP}$ with perfect completeness (i.e., zero completeness error) and negligible soundness~\cite{gmw91}. 
\end{enumerate}
\end{fact}

\begin{remark}
    Although the referenced constructions of secret-key encryption, commitments, and zero-knowledge proofs do not explicitly analyze quantum adversaries, it is straightforward to see that the computational security properties are inherited directly from the underlying one-way function. When the one-way functions are post-quantum secure, the security holds against QPT algorithms.
\end{remark}



\subsection{Position Verification}

We present the formal definition of position verification protocols for arbitrary spacetime regions $R \subseteq \R^d \times \R$ 
(which generalizes the standard notion of position verification). 

\begin{definition}[Position Verification]\label{def:pv}
    Let $V=\{V_1,\dots,V_k\}$ be a set of verifier algorithms. Let $ R \subseteq \R^d \times \R$ be a set called the \textbf{allowable prover spacetime positions}, and 
    $P=\{P_\alpha\}_{\alpha\in R}$ be a family of honest prover algorithms such that $\pos_t(P_\alpha)=L$ for each allowed spacetime point $\alpha=(L,t)\in R$. Then $\Pi=(P,V,R)$ is a \textbf{position verification protocol} with completeness error $c(\secparam)$ and position security $s(\secparam)$ against a class $\mathscr{C}$ of spoofers if the following are true.
    \begin{enumerate}
        \item \textbf{(Completeness)} For all allowed points $\alpha \in R$,
        \[
        \Pr \left [(P_\alpha\interact  V)(1^\secparam)  \text{ accepts} \right] \geq 1 - c(\secparam).
        \]
        \item \textbf{(Position Security)} For all prover coalitions $P^*=\{P^*_1,\dots,P^*_m\}\in\mathscr{C}$ where $\pos(P^*_i)\not\in R$ for all $i$, it holds that  
        \[
        \Pr \left [(P^*\interact  V)(1^\secparam)  \text{ accepts} \right] \leq s(\secparam)~.
        \]
    \end{enumerate}
    If $|R|=1$, we call $\Pi$ a \textbf{singleton} position verification protocol.
     
\end{definition}
We now introduce a special subclass of position verification protocols that will be needed for our construction of position commitments and zero-knowledge position proofs.
\begin{definition}[Nice Protocols]
\label{def:one-shot-mostly-classical}
    We say that a family of singleton position verification protocols $\{\Pi^{\mathrm{single}}_i\}_{i\in I}$ is \textbf{nice} if the following conditions are true:
    \begin{enumerate}
        \item \textbf{(One-shot, classical-response)} Each $(P=\{P_{(L,t)}\},V,R=\{(L,t)\})=\Pi^{\mathrm{single}}_i\in \Pi$ has the following structure:
        \begin{itemize}
    
            \item At the start of the protocol, the verifiers $V_1,\dots,V_k$, using a shared random string $s$, jointly prepare a quantum state $\ket{\psi}$ (which might be entangled between all the verifiers, or could only be possessed by a subset of verifiers) and a set of classical messages $(x_1,\ldots,x_k)$ which is a deterministic function of $s$. Then, each verifier $V_i$ sends a message containing their share of $\ket{\psi}$, along with $x_i$, to the purported prover at the spacetime point $(L,t)$.
            \item When the honest prover $P_{(L,t)}$ receives the messages $x_1,\ldots,x_k$ as well as the state $\ket{\psi}$, it performs some measurement on $\ket{\psi}$ based on $x_1,\ldots,x_k$, and sends the classical measurement result $y$ to all the verifiers.
            
            \item Let $t_1,\dots, t_k$ be the clock times that $y$ is received by the verifiers $V_1,\ldots,V_k$, respectively. At the end of the protocol, the verifiers collectively compute some deterministic predicate $W(s,y,t_1,\dots,t_k)$ and accept if and only if the output is $1$.
        \end{itemize}

        \item \textbf{(Coverage)} For every finite spacetime region $S\subset\R^d\times\R$, there exist spatial points $X^S_1,\dots,X^S_k\in\R^d$ whose convex hull contains $S$, and which additionally satisfy the following. For every point $\alpha\in S$, there is a protocol $(P,V,R=\{\alpha\})\in \Pi$ in which $V_1,\dots,V_k$ are positioned statically\footnote{The fact that the verifiers are stationary is not used explicitly, only the fact that the positions at each time are publicly known. We do not, therefore, rule out the possibility of accomplishing our protocols with mobile verifiers (e.g. medium-earth orbit satellites, like GPS)} at $X^S_1,\dots,X^S_k$, respectively.
    \end{enumerate}
\end{definition}

For conciseness of notation, we will typically use the singular form (a nice position verification protocol) to mean a nice family of position verification protocols.

Although \Cref{def:one-shot-mostly-classical} clearly constrains the space of all position verification protocols, we argue that it actually describes a very natural structure which occurs in many of the position verification schemes from prior literature. The first condition of \Cref{def:one-shot-mostly-classical} describes this natural structure. As for the second condition, it has been widely noted by prior work that a necessary condition for secure position verification is that the target point be located within the \emph{convex hull} of the verifiers' positions. It is also the case, in all prior position verification protocols we have come across, that this is a sufficient condition. 

For example, the $f$-BB84 protocol described in the introduction, and illustrated in \Cref{fig:1d-pos-ver}, is structured in the form of a nice position verification protocol. We will state this more concretely.
\begin{remark}[$f$-BB84 is nice]
    For all bounded spacetime regions $S\subset\R^d\times\R$, let $X^S_1,\dots,X^S_{d+1}\in\R^d$ be verifier positions whose convex hull contains $S$. Then, for each $\alpha\in S$, let $(P_\alpha,V_\alpha)$ be the $f$-BB84 prover and verifier algorithms for $\alpha$, respectively, where the verifiers are positioned at $X^S_1,\dots,X^S_{d+1}$. Then $\{(P_\alpha,V_\alpha,\{\alpha\})\}_{S\subset \R^d\times\R,\alpha\in S}$ is nice family of singleton position verification protocols. The protocol is perfectly complete, and has position security $1-\varepsilon$ against the class of adversaries who pre-share $O(\log n)$ qubits of entanglement, for some constant $\varepsilon>0$~\cite{bluhm2022single}. Position security can be strengthened to vanish asymptotically via parallel repetition, which was shown in \cite{escolafarras2025quantum}.
\end{remark}
As mentioned in the introduction, \Cref{def:one-shot-mostly-classical} also applies to the protocols of \cite{liu2022beating} and \cite{unruh2014quantum}.

\section{Position Commitments}
\label{sec:position-commitments}

In this section we formally define position commitments, and present our main construction. 

\begin{definition}[Position Commitment]\label{position-commitment}
    Let $V = \{V_1,\ldots,V_k\}$ be verifier algorithms, let $S\subset\R^d\times\R$ be a finite set called the committable points, and let $P=\{P_\alpha\}_{\alpha\in S}$ be a family of prover algorithms. Then, the tuple $(P,V,S)$ is a \textbf{position commitment scheme} if it satisfies the properties below. The interaction between $P$ and $V$ is split into two phases, $\Commit$ and $\Reveal$, which have the following syntax.
    \begin{enumerate}
        \item The input to the $\Commit$ phase is the security parameter $1^\secparam$. The $\rm Commit$ phase takes place in some time window $[t_{\rm init},t_{\rm final}]$, where $t_{\rm init}\leq t_{\rm min}(S)\leq t_{\rm max}(S)\leq t_{\rm final}$. At time $t_{\rm final}$, the verifiers jointly output a state $\rho$, called the \textbf{commitment state}.
        
        \item The input to the $\Reveal$ phase is the security parameter $1^\secparam$, the commitment state $\rho$, and a claimed position $\alpha^*\in S$. The output of the $\Reveal$ phase is the verifiers' accept or reject decision.
    \end{enumerate}
    Let $\Commit_{P \interact V}$ and $\Reveal_{P\interact V}$ denote the interaction between prover $P$ and verifiers $V$ in the $\Commit$ and $\Reveal$ phases, respectively. Let $\mathscr{C}$ be some class of spoofing provers. A position commitment scheme has \textbf{completeness error} $c(\secparam)$, \textbf{statistical position binding} $s(\secparam)$ against $\mathscr{C}$, and \textbf{computational honest-verifier hiding} if it satisfies the following properties: 
    \begin{enumerate}
        \item (\emph{Completeness}) 
        For all $\secparam \in \N$, for all $\alpha \in S$,
        \[
            \Pr \Big [ \Reveal_{P_\alpha \interact V}(1^\secparam,\rho,\alpha) \text{ accepts} : \rho \leftarrow \Commit_{P_\alpha \interact V}(1^\secparam) \Big ] \geq 1 - c(\secparam)~.
        \]
        \item (\emph{Statistical Position Binding}) For all prover sets $P^*\in\mathscr{C}$, with probability at least $1-s(\lambda)$ over $\rho \leftarrow \Commit_{P^* \interact V}(1^\secparam)$, there exists a spacetime point $\alpha_\rho$ such that two properties hold\footnote{In simpler terms, the position binding property essentially forces the adversary to (1) commit to a unique spacetime point, and (2) have actually occupied this point, in order to have any hope of convincing the verifiers in the $\Reveal$ phase.}:
        \begin{enumerate}
            \item Some prover in $P^*$ occupied the point $\alpha_\rho$.
            \item For all other points $\alpha'\neq\alpha_\rho$, and all (computationally unbounded) algorithms $A$, \[\Pr\left[\Reveal_{A\interact V}(1^\secparam,\rho,\alpha')~\mathrm{accepts}\right]=0.\] 
        \end{enumerate}

        \item (\emph{Computational Honest-Verifier Hiding}) There exists a QPT simulator $\Sim$ such that for all $\alpha\in S$, and all times $t_{\rm init}\leq \tau\leq t_{\rm final}$,
            \[ \Big \{ \Sim(1^\secparam,\tau) \Big \}_{\secparam \in \N} \approx_c \Big \{ \View_V^\tau(P_\alpha \interact  V)(1^\secparam) \Big \}_{\secparam \in \N} .\]
    \end{enumerate}
    

    Additionally, we say that a commitment scheme is \textbf{mostly-classical} if it has the following two properties. (1) The output of the $\rm Commit$ phase is classical, and (2) the $\rm Reveal$ phase is a classical protocol where no random coins are flipped by the verifiers.
\end{definition}
\begin{remark}
    Note that the position binding property serves two simultaneous roles, one which is like classical binding security of commitments, and one which is like the position security of position verification protocols. We could imagine separating position binding into a binding property and a position security property. An example of when we might want this is if a pair of provers at locations $\alpha_1$ and $\alpha_2$ run the position commitment scheme to create a commitment $\rho$, such that ${\rm Reveal}(1^\secparam,\rho,\alpha_1)$ and ${\rm Reveal}(1^\secparam,\rho,\alpha_2)$ both accept. This breaks the above definition of position binding, but in some sense the provers did not lie about their positions -- there could be an application where we even \emph{prefer} to give the provers this flexibility. We think it may be interesting for future work to consider variants of this, where the two roles are separated.
\end{remark}

\subsection{Constructing Secure Position Commitments}
\label{sec:pc-construction}

\begin{algorithm}[p]
\caption{
 Encrypt-Then-Verify Position Commitment Protocol. 
 }
 Let $W_{\alpha}$ be the verification predicate for $\Pi_{\alpha}$, as described in \Cref{def:one-shot-mostly-classical}. Let $(\Setup,\Com)$ be a classical commitment scheme, and let $(\Gen,\Enc,\Dec)$ be a secret-key encryption scheme (each with post-quantum security). We construct a position commitment scheme $\mathcal{C}=(P,V,S)$ as follows.
\begin{algorithmic}[1]
\label{def:pc-construction}
    \Procedure{Commit}{$1^\lambda$}
        \State Some verifier $V^*$ is chosen as a coordinator. At time $t_{\rm init}$, $V^*$ broadcasts $\pp\gets\Setup_\Com(1^\lambda)$.
        \State After receiving $\pp$, $P$ generates $\sk\gets \Gen(1^\secparam)$ and $c\gets\Com(\pp,\sk,r)$ with $r\gets\{0,1\}^\lambda$.
        \State Let $T$ be the maximum time a message takes to travel between any point in $S$ and any verifier $V_i$. $P$ sends $c$ to each verifier, with each message timed to arrive at $t_1=t_{\rm init}+2T$.
        \State Immediately after sending $c$, $P$ runs the following:
        \State Choose a time $t^*$ such that $\pos_{t^*}(P)=L^*$ and $(L^*,t^*)\in S$.
        \For{each $\alpha\in S$ \textbf{in parallel}}
            \If{$\alpha=(L^*,t^*)$}
                \State At time $t^*$, receive challenge $x_1\dots x_k,\ket{\psi}$ from the verifiers.
                \State Measure $\ket{\psi}$ according to $x_1,\dots,x_k$ to produce $y$, then broadcast $(\alpha,{\rm Enc}(\sk,y))$.
            \Else~(when $\alpha\neq(L^*,t^*)$)
                \For{each $i$ in $1\dots k$}
                    \State Compute the time $t_i$ that $V_i$ expects to receive the prover's response in $\Pi_{\alpha}$.
                    \State Send $(\alpha,{\rm Enc}(\sk,\bot))$ directionally to $V_i$, such that it arrives at time $t_i$.
                \EndFor
            \EndIf
        \EndFor
        \Statex
        \State Meanwhile, $V_1,\dots,V_k$ jointly run the following program, starting at time $t_1$:
        \For{each $\alpha\in S$ \textbf{in parallel}}
            \State Take some shared randomness $r_{\alpha}$ and use it to prepare a challenge $\ket{\psi},x_1,\dots,x_k$
            \State Each $V_i$ sends $x_i$, and its respective portion of $\ket{\psi}$, to the spacetime point $\alpha$.
            \State Each verifier timestamps and records all prover messages to their transcript.
        \EndFor
        \State Let $t_{\rm final}$ be the latest arrival time of any prover message. Let $M$ be the combined transcript of prover messages received by the verifiers, and let $s$ be the verifiers' entire shared randomness.
        \State At $t_{\rm final}$, the verifiers output $\rho=(\pp,c,M,s)$
    \EndProcedure
    \Statex
    \Procedure{Reveal}{$1^\lambda,\rho=(\pp,c,M,s,(L^*,t^*)$}
        \State $P$ sends $(L^*,t^*)$, along with $\sk$ and $r$, to $V^*$.
        \State $V^*$ checks $c=\Com(\pp,\sk,r)$; reject if invalid.
        \State Verifiers create an empty list $A$ of accepting positions.
        \For{$\alpha\in S$}
        \State Each $V_i$ extracts the prover message labeled $(\alpha,\dots)$ from their transcript of $\rm Commit$, letting $t_i$ be its timestamp, and decrypts it with ${\rm Dec}(\sk,\cdot)$. If any messages are $\bot$, or if the decryption differs between verifiers, \textbf{continue} to the next $\alpha$. Otherwise, let $y$ be the decryption.
        \State Verifiers compute $W_{\alpha}(s_{\alpha},y,t_1,\dots,t_k)$. If it equals 1, then they add $\alpha$ to $A$.
        \EndFor
        \State Accept if $A=\{(L^*,t^*)\}$, and reject otherwise.\label{line:pc-accept-check}
    \EndProcedure
\end{algorithmic}
\end{algorithm}

\begin{lemma}\label{thm:construction-secure}
    Assume there exists a nice family of singleton position verification protocols $\{\Pi_\alpha=(P_\alpha,V_\alpha,\{\alpha\})\}_{\alpha\in S}$ for some finite spacetime region $S$, with completeness $c(\secparam)$ and position security $s(\lambda)$ against spoofers who pre-share at most $E(\lambda)$ qubits of entanglement. Assume that post-quantum one-way functions exist. Then the construction in \Cref{def:pc-construction}, instantiated with the aforementioned primitives, is secure under \Cref{position-commitment} with completeness $c(\lambda)$, statistical position binding $|S|(s(\lambda)+\negl(\lambda))$ against spoofers sharing $E(\lambda)/2$ bits of entanglement, and computational honest-verifier hiding. Additionally, $\mathcal{C}$ is mostly-classical.
\end{lemma}
\begin{proof} We first show that $\mathcal{C}$ is mostly-classical, which follows directly from the ``one-shot, classical-response'' property of the nice protocol $\Pi$, as well as the definitions of the honest prover and verifiers which do not introduce any quantum states, or flip any random coins in $\Reveal$. Now we prove the rest of the claim:
    \paragraph{Completeness:} Take an honest prover $P_\alpha$. At some point during $\rm Commit$, $V$ and $P_\alpha$ will execute the protocol $\Pi_\alpha$, except that $P_\alpha$'s messages are encrypted. Since $\Pi_\alpha$ has the structure of a nice protocol, however, we can argue that by decrypting the prover's messages at reveal time and resuming verification, we reach the same state as an unmodified execution of $\Pi_\alpha$. Thus the verifiers accept with probability $1-c(\lambda)$.
    \paragraph{Computational Honest-Verifier Hiding:} For hiding, we observe that no matter where $P$ is located, the transcript $M$ of messages each verifier $V_i$ during the $\rm Commit$ phase \emph{always} consists of receiving the commitment $c$, and then a set of encrypted messages from the prover whose timings are consistent with $\Pi_\alpha$, for every point $\alpha\in S$. This is made possible by the fact that $P_\alpha$ knows the precise time that each verifier expects each message in $\Pi_{\alpha'}$ for all $\alpha'\in S$, and has a sufficient head start on the verifiers to be able to send all of these messages in time. For any public parameters $\pp$, we can define a simulated transcript $M'_\pp$ as follows: take $\sk\gets\Gen(1^\secparam)$, $r\gets\{0,1\}^n$, and $c\gets\Com(\pp,\sk,r)$. Add $(t_1,c)$ to $M'_\pp$, where $t_1=t_{\rm init}+2T$ (recall that $T$ is the maximum time to reach any verifier from any point in $S$). Then add ``$V_i\gets(\alpha,{\rm Enc}({\rm sk},\bot))~@~t$'' to $M'_\pp$ for every time $t$ where each verifier $V_i$ would receive a message in $\Pi_\alpha$ for every $\alpha\in R$. Now, we define the simulator $\Sim$ to compute the following on input $(1^\secparam,\tau)$: Start simulating $V$, which outputs $\pp$ at time $t_{\rm init}$. Construct the timestamped transcript $M'_\pp$ as defined above. Then, simulate $V$ through time $\tau$, while feeding it messages exactly according to $M'_\pp$. Then, output $V$'s internal state.
    
    Once we are convinced that the simulated transcript matches the real transcript in terms of quantity and timing of messages, it follows straightforwardly from the hiding property of the classical commitment scheme, as well as the security of the secret key encryption scheme, that the output of the simulator is computationally indistinguishable from the verifiers' view during the honest commitment protocol.


    \paragraph{Statistical Position Binding:}
    Take a prover strategy $P^*$ which pre-shares at most $E/2$ qubits of entanglement.
    Assume that $P^*$ breaks statistical position binding with probability $p$ -- i.e. with probability $p$ over $\rho\gets\Commit_{P^*\interact V}$, one of conditions (a) and (b) in the definition is not satisfied. Because $\rm (Setup,Com)$ is statistically binding, $P^*$ can only distinguish it from a perfectly binding scheme with negligible probability. Thus we can take an equivalent view which removes this event, and where we adjust all probabilities accordingly by a negligible amount. 
    
    So, we henceforth treat $\rm (Setup,Com)$ as perfectly binding, and say that $P^*$ breaks statistical position binding with probability at least $p-\negl(\lambda)$. Under this perfect binding, there cannot be two points $\alpha_1$ and $\alpha_2$ such that $\rho$ can open to both, since decryption under $\sk$ is deterministic and due to the check for a unique point on line \ref{line:pc-accept-check} of \Cref{def:pc-construction}. This means that with probability $p-\negl(\secparam)$, $P^*$ breaks position security by \emph{spoofing its position}. In other words, with probability $p-\negl(\secparam)$ over $\rho\gets\Commit_{P^*\interact V}$, there exists some $\Reveal$ strategy which succeeds in revealing $\rho$ at a point that $P^*$ never occupied. We can break this down into point-by-point cases: for any point $\alpha\in S$ occupied by no provers in $P^*$, let $q_\alpha$ be the probability that there exists some successful strategy for revealing to $\alpha$. By a union bound over points in $S$, we can take some maximizing choice $\alpha^*$ such that $q_\alpha^*\geq (p-\negl(\lambda))/|S|$.
    Let $q=(p-\negl(\lambda))/|S|$ for simplicity of notation. To restate the above, we have found a point $\alpha^*$ which is not occupied by any prover in $P^*$, such that commitments generated by $P^*$ can be opened to position $\alpha^*$ with probability at least $q$.

    Now by averaging, choose some $\rm pp^*$ and $c^*$ such that when we condition on $\pp=\pp^*$ and $c=c^*$ (the public parameters and commitment to the secret key, respectively) in the $\rm Commit$ phase, $\rho$ can still be opened to $\alpha^*$ with at least the same probability $q$. Let ${\rm sk},r$ be the unique values such that $c^*\gets\Com(\pp,{\rm sk},r)$ -- recall, we are treating $\Com$ as perfectly binding. Given this choice, let $\omega$ be the (possibly entangled) internal state across all provers in $P^*$ just after the commitment $c$ was announced in the $\rm Commit$ phase, conditioned on $c=c^*$ and $\pp=\pp^*$\footnote{This can be thought of as a quantum mixed state across the combined registers of each prover.}. We now describe the reduction, which takes $\omega,{\rm sk}$ as nonuniform advice:

    For all $i$, let $\Tilde{P}_i$ be a reduction algorithm positioned exactly as $P^*_i$ was in the $\rm Commit$ phase of the original cheating strategy -- that is, for all $t_{\rm init}\leq t\leq t_{\rm final}$, $\pos_t(\Tilde{P}_i)=\pos_t(P^*_i)$. $\Tilde{P}_i$ controls an instance of $P^*_i$, in the sense that it can read and write to $P^*_i$'s internal state and can intercept messages which are sent and received by $P^*_i$. We now have $\Tilde{P}$ participate in an execution of the vanilla position verification experiment $\Pi_{\alpha^*}$, with verifiers $\Tilde{V}$, whose positions are identical to the verifiers $V$ from the commitment protocol. At the start of the experiment, the provers in $\Tilde{P}$ set the joint state across the registers of all of the simulated provers in $P^*$ to the advice state $\omega$. The verifiers execute $\Pi_{\alpha^*}$, and the provers begin simulating the execution of the cheating strategy $P^*$. Whenever some prover $\Tilde{P}_i$ detects that its simulation of $P^*_i$ is sending a message $(\alpha,m)$ and $\alpha=\alpha^*$, it intercepts this and computes $m'={\rm Dec}({\rm sk},m)$. $\Tilde{P}_i$ then forwards $m'$ to the intended verifier.

    We now analyze the success probability of this reduction. In the commitment experiment, in order for $P^*$ to have generated a commitment which can open to $\alpha^*$ after reaching the state $\omega$, it must have sent a set of messages of the form $(\alpha^*,\Enc(\dots))$ during the $\rm Commit$ phase, whose decryptions under $\rm sk$ would cause the verifiers for $\Pi_{\alpha^*}$ to accept. This is clear from the verification process in the protocol, along with the fact that the provers had output $c^*$ which is perfectly binding to $\rm sk$. Now since $\Tilde{P}$ sends exactly $P^*$'s messages but decrypted under $\rm sk$, then $\Tilde{V}$ must accept whenever $P^*$ would have generated such a commitment. Thus, the reduction convinces the verifiers $\Tilde{V}$ with exactly probability $q$. 

    Now, note that after we added the advice state $\omega$ to the cheating strategy $\Tilde{P}$, the amount of preshared entanglement increased to be at most $E$ qubits. Since we assumed that $\Pi_{\alpha^*}$ has position security $s(\secparam)$ against this class of prover, we have shown that our position commitment protocol has position binding $|S|( s(\secparam)+\negl(\secparam))$ against spoofers with $E/2$ qubits of preshared entanglement.
\end{proof}
\subsection{Relationship to Other Quantum Cryptographic Primitives}
\begin{proposition}\label{thm:equivalence}
    Assume there exists a position commitment scheme $\mathcal{C}=(P,V,S)$ with completeness $c(\secparam)$, statistical position binding $s(\secparam)$ against some class $\mathscr{C}$ of adversaries, and computational hiding security. Then, the following hold:
    \begin{enumerate}
        \item There exists a (not necessarily nice) singleton quantum position verification protocol for each point in $S$, with completeness $c(\secparam)$ and with position security $s(\secparam)$ against $\mathscr{C}$.
        \item If there are no computational limits on $\mathscr{C}$, then there exist quantum bit commitments with completeness $c(\secparam)$, $s(\secparam)$-statistical binding, and honest-verifier computational hiding against semi-honest receivers.
    \end{enumerate}
\end{proposition}
\begin{proof}
We will start with point 1. The position verification protocol $\Pi=(P,V,\{(L,t)\})$ is defined as follows: $P$ and $V$ run the $\Commit$ phase of $\mathcal{C}$, and then proceed immediately to the $\Reveal$ phase. The hiding property is no longer important, but it is immediately clear that the completeness and position binding of $\mathcal{C}$ are inherited directly as the completeness and position security of $\Pi$, respectively.

We move on to point 2, constructing an interactive protocol for (semi-honest receiver) quantum commitments. Taking some arbitrary spatial point $L$, designate two fixed spacetime points $(L,t_0),(L,t_1)\in S$ to represent bits 0 and 1, respectively. To commit to $b\in\{0,1\}$, the sender (who we will call $P$, since it takes the place of the prover) will engage the receiver (who we will call $V$, since it takes the place of the verifier) in a simulated version of the $\Commit$ phase of $\mathcal{C}$, where it commits to its position $(L,t_b)$. The simulation will take place in the (untimed, un-positional) interactive-protocol setting as follows. There is a round of interaction for each timestep $t$ of the $\Commit$ phase:
\begin{itemize}
    \item $V$ sends $P$ any messages \textbf{which would arrive} at $(L,t)$.
    \item $P$ sends $V$ any messages \textbf{which it would send} at time $t$.
\end{itemize}
The output of this protocol is the commitment state $\rho$. To reveal, the parties simply run $\Reveal$ in the straightforward way.

We now argue for the completeness and hiding properties. The reason for the asymmetry, where $V$ sees $P$'s messages slightly \emph{ahead} of schedule, is because we assume $V$ behaves semi-honestly and so cannot use this information to distort the protocol. Notice, then, that $P$'s view at any time $\tau$ is identical to what it would have been in the protocol -- this is true no matter which bit $P$ commits to, since it's simulating a prover at the same spatial position in either case. Given that the prover's final view is identical to that of $\Commit$, it also follows that the verifier's final view is identical to that of $\Commit$. From this it is easy to see that the protocol has completeness $c(\secparam)$. Additionally, the verifier's view is trivially simulatable by removing the timestamps from the simulator provided by the honest-verifier computational hiding security of $\mathcal{C}$. The semi-honest-receiver computational hiding property follows.

The statistical binding property follows in a similar fashion. If it was violated, we could construct a prover for $\mathcal{C}$ which breaks position binding with the same probability, by sending exactly the messages that are sent in the simulation above.
\end{proof}
In the case of quantum commitments, luckily, we can leverage some nice equivalence results to strengthen the properties relative to what we showed above.
\begin{theorem}[\cite{yan22general}]
    Any (interactive) quantum bit commitment scheme whose purification is semi-honest secure can be converted into a non-interactive quantum bit commitment scheme with perfect completeness, and whose hiding and binding properties are exactly preserved.
\end{theorem}
In our case, the purified version would maintain hiding (since the prover keeps the purified secret key privately) and binding (purification cannot help break classical (post-quantum secure) commitments in any way), so we can apply the theorem.
\begin{corollary}
    Take the assumptions in \Cref{thm:equivalence}. Then, there exists a non-interactive quantum bit commitment scheme with perfect completeness, $s(\secparam)$-statistical binding, and computational hiding.
\end{corollary}
\subsection{Optimizations for Concrete Constructions}
\label{sec:optimization}
In \Cref{sec:pc-construction}, we showed how to construct position commitments from any nice position verification protocol. However, 
our construction is quite computationally demanding on both the prover and verifier algorithms, which puts some strain on our modeling assumption of instantaneous computation. In particular, it seems challenging for the algorithms to run a position verification protocol for every point in the committable set \emph{simultaneously}. We would prefer a solution which only requires a small, constant amount of computation per timestep. 

We will now informally describe an optimized construction which can hopefully serve as evidence for the real-world implementability of position commitments and, subsequently, zero-knowledge position proofs. We will require a couple of additional properties from the underlying position verification protocol, on top of being nice: the verifiers' messages must be \textbf{1.} entirely classical, and \textbf{2.} generated independently of one another. We note that this is satisfied, for example, by the protocol of \cite{liu2022beating}\footnote{The main protocol of \cite{liu2022beating} is only discussed in the 1-D setting, but they give another protocol in higher dimensions which is conjectured to be secure, and also satisfies the constraints of our optimization}. The first property is so that the messages can be broadcast, instead of directional (since there is no meaningful notion of ``broadcasting'' a quantum state). The second property allows the verifiers to build a mesh of challenges throughout the spacetime region, as pictured in \Cref{fig:2d-pc-optimized}. Wherever a prover is, there will be some combination of independently sampled challenges arriving at (or very near) that spatial point at some time, which the prover can then respond to under encryption as before. As for the prover's task of pretending to be in all other points at once: observe that as long as the verifiers to receive at least one encrypted message at each timestep of the protocol, they cannot rule out the presence of a prover at any given position. On the other hand, a prover can easily induce such a state by sending each verifier one encrypted message per timestep (where, similarly to \Cref{def:pc-construction}, it sends directional messages delayed appropriately so that all verifiers receive the first one and the last one at the same time). The optimizations are informally described below.
\theoremstyle{plain}
\newtheorem{protocolx}[algorithm]{Protocol}
\begin{protocolx}[Optimized Encrypt-Then-Verify Scheme]
\label{def:optimized-pc}
    We define this protocol informally by specifying the main loop of the provers and verifiers, and where everything outside of this loop proceeds in roughly the same way as \Cref{def:pc-construction}.
    \begin{itemize}
        \item The verifiers do the following \emph{at each timestep $t$}: each verifier $V_i$ generates a challenge string $x_{i,t}\in\{0,1\}^n$, chosen from the verifiers' shared randomness, and broadcasts it.
        \item The honest prover at location $L$, at each timestep $t$, checks whether it received a set of verifier challenges $(x_1,\dots,x_{d+1})$. If not, $P$ sends $\Enc(\bot)$ to each verifier. If it did, then $P$ computes some output $y(x_1,\dots,x_{d+1})$ and sends ${\rm Enc}({\rm sk},y)$ to each verifier. The prover delays these messages appropriately so that they are received in lockstep by all verifiers simultaneously.
    \end{itemize}
\end{protocolx}
We claim that \Cref{def:optimized-pc} is complete for all spacetime points where the verifiers' broadcast messages intersect. We provide a basic visual intuition in \Cref{fig:2d-pc-optimized}, and it is not hard to argue intuitively that, with a message broadcast at every timestep there is a fine mesh of committable points induced over the convex hull of the verifiers' positions. The arguments for hiding and binding of this commitment follow essentially in just the same way as \Cref{thm:construction-secure} -- the verifiers' view is always a constant stream of encrypted messages, and the prover is bound by its commitment to the secret key and by the security of the underlying position verification scheme. As for the optimization: by observing the algorithm for honest parties, it is clear that the \emph{per-timestep} computational load on all parties in this protocol is a quantity which does not grow with the number of committable points $|S|$, and is polynomial in the security parameter.
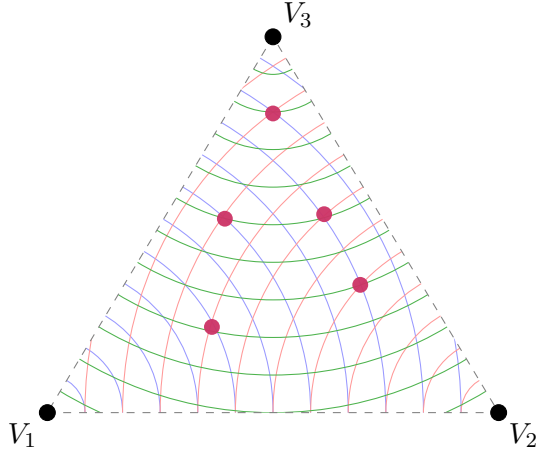
\begin{figure}[ht]
    \centering
    \begin{tikzpicture}[scale=2.0]
        \coordinate (A) at (0,0);
        \coordinate (B) at (3,0);
        \coordinate (C) at (1.5, 2.5);

        \begin{scope}
            \clip (A) -- (B) -- (C) -- cycle;

            \foreach \r in {0.25, 0.5, ..., 3.25} {
                \draw[blue!40, thin] (A) circle (\r);
            }
            \foreach \r in {0.25, 0.5, ..., 3.25} {
                \draw[red!40, thin] (B) circle (\r);
            }
            \foreach \r in {0.25, 0.5, ..., 3.25} {
                \draw[green!40!gray, thin] (C) circle (\r);
            }
        \end{scope}

        \draw[gray, dashed, thin] (A) -- (B) -- (C) -- cycle;
        
        \node[below left] at (A) {$V_1$};
        \node[below right] at (B) {$V_2$};
        \node[above right] at (C) {$V_3$};
        \filldraw[black] (A) circle (1.5pt);
        \filldraw[black] (B) circle (1.5pt);
        \filldraw[black] (C) circle (1.5pt);

        \fill[purple!80, opacity=0.95] (1.094, 0.57) circle (1.5pt);
        \fill[purple!80, opacity=0.95] (1.180, 1.29) circle (1.5pt);
        \fill[purple!80, opacity=0.95] (1.500, 1.99) circle (1.5pt);
        \fill[purple!80, opacity=0.95] (1.84, 1.32) circle (1.5pt);
        \fill[purple!80, opacity=0.95] (2.08, 0.85) circle (1.5pt);

    \end{tikzpicture}
    \caption{Illustration of the optimized position commitments in \Cref{def:optimized-pc} for $d=2$, with intersection points highlighted in red. When these arcs are as dense as one per timestep, the set of intersection points will be a fine-grained mesh.}
    \label{fig:2d-pc-optimized}
\end{figure}
\section{Zero-knowledge Position Verification}\label{sec:zklp}
As outlined in \Cref{sec:zklp-intro}, the main goal of this paper is to construct secure position proofs which are zero-knowledge over some nontrivial region of spacetime, in the following sense:
\begin{definition}[Zero-Knowledge Property for Position Verification]
    Let $\Pi=(P,V,R)$ be a position verification protocol, where $P$ is a family $\{P_\alpha\}_{\alpha\in R}$ of honest prover algorithms. We say that $\Pi$ has the \textbf{honest-verifier, computational zero-knowledge} property over $R$ if the following is true. There exists a QPT simulator $\Sim$ such that for any $\alpha\in R$, and any time $\tau$,
    \[ \Big \{ \Sim(1^\secparam,\tau) \Big \}_{\secparam \in \N} \approx_c \Big \{ \View^\tau_V(P_\alpha \interact  V)(1^\secparam) \Big \}_{\secparam \in \N} \]
\end{definition}
We notice that the zero-knowledge property is very similar to the hiding property of position commitments. Naturally, we will see that we can construct zero-knowledge position proofs by proving statements about the potential reveal outcomes of a position commitment.
\begin{lemma}\label{thm:zklp-construction}
    Suppose $\mathcal{C}=(P,V,S)$ is a mostly-classical 
     position commitment scheme with completeness error $c(\lambda)$, statistical position binding $s(\lambda)$ against some prover class $\mathscr{C}$, and computational honest-verifier hiding. Then, for any nonempty region $R\subset S$, there exists a position verification protocol $\Pi=(P,V,R)$ with completeness $c(\lambda)$, position security $s(\lambda)+\negl(\lambda)$ against $\mathscr{C}$, and honest-verifier computational zero knowledge.
\end{lemma}
\begin{proof}
    Take such an $R\subset S$, and fix any $\alpha^*\in R$. Let $P_\mathcal{C}$ be some honest prover for $\mathcal{C}$ located at $\alpha^*$. We now describe the zero-knowledge position verification protocol, by defining the behavior of the honest prover $P=P_{\alpha^*}$ and the honest verifiers $V$.
    
    $P$ and $V$ first run $\rho\gets\rm Commit_{P_\mathcal{C}\interact V}(1^\lambda)$, generating a classical commitment state which $V$ then sends to the prover. $P$ will then initiate a zero-knowledge proof that $\rho$ belongs to an NP language ${\sf Reveal}_R$, which we define below. 
    \[\rho\in{\sf Reveal}_R\iff\exists(\alpha,P^*)\ s.t.\ \alpha\in R\ \land\ P^*\text{ is a PPT alg.}\ \land\ \Reveal_{P^*\interact V}(1^\lambda,\rho,\alpha)\text{ accepts}\]
    We know that under our assumptions, there exists a post-quantum computational zero-knowledge proof protocol ${\rm ZK}=(P_{ZK},V_{ZK})$ for ${\sf Reveal}_R$ with perfect completeness and negligible soundness error. So the protocol $\Pi$ continues as follows: $P$ and $V$ run $P_{ZK}$ and $V_{ZK}$, respectively, on common input $\rho$. $P_{ZK}$ is additionally given a witness $w=(\alpha^*,P_\Reveal)$, where $P_\Reveal$ is $P$'s code for the $\Reveal$ phase. Afterwards, $V$ accepts if and only if $V_{ZK}$ accepts.

    We claim that $\Pi$ has completeness error $c(\lambda)$: this follows from the fact that $\mathcal{C}$ has completeness error $c(\lambda)$ and $\rm ZK$ has perfect completeness. Next, we argue for $\Pi$'s position security: take a cheating prover coalition $P^*\in\mathscr{C}$ where nobody occupies position $\alpha^*$. By the statistical position binding of $\mathcal{C}$, except with probability $s(\lambda)$ over $\rho\gets\rm Commit_{P^*\interact V}(1^\lambda)$, we have that $\Reveal_{A\interact V}(1^\lambda,\rho,\alpha)$ rejects for all points $\alpha$ and all algorithms $A$. Thus, $\rho$ only belongs to the language ${\sf Reveal}_R$ with probability at most $s(\lambda)$. By the soundness of $\rm ZK$, this means that $V$ accepts with probability at most $s(\lambda)+\negl(\lambda)$.

    Finally, we describe the zero-knowledge simulator $\Sim$. Take ${\Sim}_\mathcal{C}$ to be the simulator guaranteed by the computational hiding property of the commitment scheme, and ${\Sim}_{\rm ZK}$ be the simulator for the protocol $\rm ZK$. On input $(1^\lambda,\tau)$, $\Sim$ will first run $\sigma={\Sim}_\mathcal{C}(1^\lambda,\min(t_{\rm final},\tau))$. If $\tau\leq t_{\rm final}$, $\Sim$ outputs $\sigma$. Otherwise, take the commitment state $\rho$ which is contained on the verifiers' output register in $\sigma$ (and let $\sigma'$ be the remainder of the verifiers' state). Next, $\Sim$ runs $\pi\gets{\Sim}_{\rm ZK}(1^\lambda,\rho)$ and outputs $(\sigma',\pi)$. As we have constructed $\Pi$, and by the assumed properties of ${\Sim}_\mathcal{C}$ and ${\Sim}_{\rm ZK}$, the output of ${\Sim}(1^\lambda,\tau)$ is computationally indistinguishable from $\View^\tau_V(P\interact V)(1^\lambda)$.
\end{proof}
\begin{theorem}[Main Result]\label{thm:main}
    Suppose there exist post-quantum one-way functions, and a nice position verification protocol with completeness $c(\secparam)$ and position security $s(\secparam)$ against spoofers sharing at most $E$ entangled qubits. Then, for any finite region $R\subset \R^d\times\R$, there exists a position verification protocol $\Pi=(P,V,R)$ with completeness $c(\secparam)$, position security $|S|(s(\secparam)+\negl(\secparam))$ against spoofers sharing $E/2$ entangled qubits, and honest-verifier computational zero knowledge.
\end{theorem}
\begin{proof}
    Let $\mathcal{C}=(P,V,S)$ be the encrypt-then-verify position commitment scheme outlined in \Cref{sec:pc-construction}, instantiated with the above assumptions. Then, the existence of $\Pi$ as claimed, follows from \Cref{thm:zklp-construction}.
\end{proof}
\section{Towards Malicious-Verifier Zero Knowledge}
\label{sec:towards-malicious}

We have focused so far on the honest-verifier setting. We feel that this is a reasonable assumption in some settings, such as for example when the verifiers are cell towers; malicious behavior would likely be detectable by, or significantly impact quality of service to, a large customer base. However, in full generality, it would clearly be desirable to achieve zero-knowledge even against malicious verifiers. With this in mind, we will give some intuition for why the honest-verifier assumption may be difficult to remove, and then give some potential directions for relaxing it.

\paragraph{A General Attack on Privacy.} If we naively try to work within a universe with malicious verifiers, we run into a simple, but annoyingly effective, ``poke-around-in-the-dark'' attack on privacy. This attack, for any positional protocol $\Pi$ with good completeness, goes as follows: the verifiers \emph{only} send messages along a particular ray $\Vec{z}\in\R^d$, and then see whether they accept or reject at the conclusion of $\Pi$. The protocol will accept with high probability if and only if the honest prover is located along $\Vec{z}$! Equivalently, we could imagine that they send messages \emph{everywhere but} a particular direction, and now they will only reject if the prover's position was in that direction. The major obstacle in defending against these attacks in our model is that a party can only see messages which pass through their location; thus, no honest party can tell whether or not messages were sent directionally to other locations
 (much like the proverbial tree falling in a forest with no one around to hear it). 

Thus we claim (informally) that for any protocol $\Pi=(P,V,R)$ where in order for completeness to hold, $V$ must send directional messages 
to every point in $R$, 
 the protocol cannot be zero-knowledge against malicious verifiers. This is due to the attack described above, where $V$ can ``deny service'' to some zone of $R$, and learn whether the prover was in this zone based on the eventual decision to accept/reject. 

\paragraph{Possible Directions}
As demonstrated above, it is clear that we cannot hope to achieve malicious-verifier privacy when the protocol hinges on the verifiers' use of directional messages. One might note, however, that our optimized protocol in \Cref{sec:optimization} utilizes verifier messages which are sent using entirely classical, broadcast signals. Note that the prover, however, still has to send directional messages to the verifiers, in order to mask their position. However, if we can reasonably assume a model where verifiers only send broadcast messages, but provers can send either directional or broadcast messages, then we could hope to prove zero-knowledge against malicious verifiers.
 The broadcast-only would be reasonable to assume if, say, 1. the verifiers are low-tech cell towers without directional capability, or even if 2. the verifiers are arbitrary, but the prover can \emph{differentiate} whether a signal it received was directional or broadcast -- if the prover simply ignores all directional messages, then we essentially collapse to the broadcast-only model.


Another interesting area of investigation is to augment the model with more verifiers. We could consider having a redundant number of verifiers in each protocol, i.e. a considerable surplus to the $d+1$ which are normally required, and use an honest-majority assumption to see what privacy guarantees can be achieved. To motivate this model, we can again use the example of cell towers and observe that points in densely populated areas are likely to be in range of dozens of cell towers at once -- much more than the four that are required to perform 3-D position verification. We could imagine that with this surplus, the prover is free to randomize the choice of verifiers which it actually listens to, and thus with high probability decouple the accept/reject decision from the actions of any given verifier (or any minority coalition of dishonest verifiers).
 We note that in this model there is also an interesting twist to be considered -- since each verifier is surrounded by other verifiers, we could ask it to play the role of \emph{a prover}, and establish some property of its position to the rest of the network. An even simpler observation about this model is that the surrounding verifiers can audit the timings of its broadcast messages, which could be useful in the aforementioned discussion about achieving malicious-verifier zero-knowledge in the broadcast-only model.



\section*{Acknowledgments} 
Diagrams in this paper were made with the help of generative AI.
\iftrue
HY is supported by AFOSR award FA9550-23-1-0363, NSF awards CCF-2530159, CCF-2144219, and CCF-2329939, and by the Sloan Foundation.
\fi

\printbibliography
\end{document}